\documentclass[lineno]{preprint}

\usepackage{amsmath}

\usepackage{amsfonts}
\usepackage{graphicx}
\usepackage{epstopdf}
\usepackage{algorithmic}
\usepackage{amsmath}
\usepackage{amssymb}
\usepackage{tikz}
\usepackage{pgfplots}
\pgfplotsset{compat=1.18}
\usepackage{newtxtext}
\usepackage[subscriptcorrection]{newtxmath}
\usepackage{float} 
\usepackage{caption}

\graphicspath{{./art/}}

\usepackage[plain,noend]{algorithm2e}

\makeatletter
\renewcommand{\algocf@captiontext}[2]{#1\algocf@typo. \AlCapFnt{}#2} 
\def\@algocf@capt@plain{top}
\renewcommand{\algocf@makecaption}[2]{%
  \addtolength{\hsize}{\algomargin}%
  \sbox\@tempboxa{\algocf@captiontext{#1}{#2}}%
  \ifdim\wd\@tempboxa >\hsize
    \hskip .5\algomargin%
    \parbox[t]{\hsize}{\algocf@captiontext{#1}{#2}}
  \else%
    \global\@minipagefalse%
    \hbox to\hsize{\box\@tempboxa}
  \fi%
  \addtolength{\hsize}{-\algomargin}%
}
\makeatother

\begin{document}


\title{A Separation Method for Quartic Positivity and the Valid Region of Gram--Charlier densities}
\author{Taehun Kim}
\affil{Interdisciplinary Program in Bioengineering, Graduate School, Seoul National University, Seoul 08826, Republic of Korea.
\email{fealty@snu.ac.kr}}

\author{Jung Chan Lee}
\affil{Department of Biomedical Engineering, Seoul National University College of Medicine and Seoul National University Hospital, Seoul 03080, Republic of Korea.\email{ljch@snu.ac.kr}}

\author{ByoungSeon Choi}
\affil{Graduate School of Data Science, Seoul National University, Seoul 08826, Republic of Korea.
\email{bschoi12@snu.ac.kr}
}

\maketitle

\begin{abstract}
The positivity of the Gram--Charlier probability density function has been a subject of extensive study for decades. Since \cite{barton1952conditions} introduced numerical positivity conditions, no analytic closed-form expression was available until \citeauthor{kwon2019positive} (\citeyear{kwon2019positive}, \citeyear{kwon2022analytic}) proposed analytic solutions for the valid region of Gram--Charlier densities. Despite the significance of the analytical solutions, the expressions remain algebraically complex. As these conditions for the Gram--Charlier densities are determined by a quartic polynomial, it is essential to investigate its positivity.
In this work, necessary and sufficient conditions for the positivity of a quartic polynomial are derived through a separation method.
Based on these conditions, more concise analytic expressions for the positivity of the Gram--Charlier density are proposed.
\end{abstract}


\section{Introduction}

While numerous natural phenomena conform to the Gaussian distribution, many domains exhibit significant deviations, especially in the tails. These heavy-tailed characteristics present difficulties for modeling extreme events. The Gram--Charlier series expansion emerges as a natural extension, effectively capturing such non-Gaussian features through higher-order corrections to the Gaussian density \citep{jondeau2001gram}. Numerous studies have investigated Gram--Charlier densities, e.g., \cite{SHENTON1951gramcharlierbiometrika}.

The Gram--Charlier series expansion employs Hermite polynomials to incorporate third-order (skewness) and fourth-order (kurtosis) information. While this enhances the ability to represent tails and skewness compared to the Gaussian density, the presence of the associated quartic polynomial introduces the possibility of non-positive values. This issue has motivated needs to address the positivity conditions of the quartic function.
\section{Quartic Positivity}
\subsection{Related Works}
The positivity conditions of a quartic polynomial have been studied for decades. See, e.g.,\ \cite{rees1922graphical}, \cite{ku1965automatic}, 
\cite{lazard1988quantifier}, 
\cite{hasan2002procedure}, \cite{jury1981positivity}, \cite{qi2020positivity}, and 
\cite{song2021analytical}. 
Even though \cite{rees1922graphical} and \cite{lazard1988quantifier} provide a closed-form expressions via four inequalities, their complexity makes practical implementation difficult. Yet symbolic elimination becomes intractable, particularly as the discriminant of the associated quartic equation manifests in a high-degree polynomial form, further complicating the positivity domain.

We introduce a novel theorem that simplifies the positivity conditions of a quartic equation, yielding a more tractable representation.

\subsection{Conditions}
Our aim is to provide a characterization of the positivity of a reduced quartic polynomial, i.e.,\ 
\begin{equation*}
\mathbf{(C0)} \quad f(x) = x^4+px^2+qx+r>0 \quad \forall x \in \mathbb{R}.
\end{equation*}
Following Ferrari's method for solving a quartic equation, published by his teacher \cite{cardano1545arsmagna}, we introduce a free parameter $m$ and define two functions as
\begin{align*}
h_m(x) &\doteq \left(x^2+\frac{p}{2}+m\right)^2 \geq 0 \quad \forall x \in \mathbb{R}, \\
g_m(x) &\doteq 2mx^2-qx+\left(m+\frac{p}{2}\right)^2-r.
\end{align*}
Clearly, 
\begin{equation*}
f(x)=h_m(x)+\left\{-g_m(x)\right\}.
\end{equation*}
\begin{figure}[htbp]
\centering
\begin{tikzpicture}
\begin{axis}[
    legend pos=south west,
    axis lines = middle,
    xlabel = \(x\),
    ylabel = \(y\),
    domain = -3:3,
    ymin = -10,
    ymax = 10,
    samples = 100,
    width=14cm,
    height=8.3cm,
]
\addplot [blue, thick] { (x^2-0.5625)^2 };
\addplot [red, thick] { -0.75*(x-0.1)^2+1.125*(x-0.1)-0.421875-0.3 };
\addplot [dash pattern=on 12pt off 4pt on 12pt off 4pt, black, line width=2pt] {0};
\node[text width=3cm,scale=2.0] at (2.0,9) {$h_m(x)$};
\node[text width=3cm,scale=2.0] at (3.0,-5) {$g_m(x)$};
\node[text width=3cm,scale=1.3] at (-2.0,1.1) {Separation line};
\end{axis}
\end{tikzpicture}
\caption{Functions $h_m(x)$ and $g_m(x)$, and the separation line.
\\
\textbf{Alt text}: Plot illustrating the separation argument: $h_m(x)$ remains nonnegative, 
$g_m(x)$ remains strictly negative, and the horizontal line $y=0$ separates the 
two functions.}
\label{fig:figure1}
\end{figure}
\\
If there exists a negative $m$ satisfying $g_m(x)<0$ for any $x \in \mathbb{R}$, then $\mathbf{(C0)}$ satisfies. At that case the horizontal line $y=0$ is a separation line between ${h_m(x)}$ and ${g_m(x)}$ as shown in Fig. \ref{fig:figure1}.
Using $h_m(x)$ and $g_m(x)$, we present some necessary and sufficient conditions for the positivity of a quartic polynomial.
\begin{theorem}\label{theorem1} The Followings Are Equivalent:
\begin{align*}
&\mathbf{(C0)} \quad \text{The reduced quartic polynomial is positive, i.e.,\ }\\
&\quad \quad \quad \; f(x) = x^4 + p x^2 + q x + r > 0 \quad \forall x \in \mathbb{R}.\\
&\mathbf{(C1)}\quad \text{There exists a negative }m \text{ such that } D(m) < 0, \text{ where }D(m) \text{ is defined as}\\
& \quad \quad \quad \; D(m) \doteq -8m\left(m^2+pm+\frac{p^2}{4}-r\right)+q^2,\\
& \quad \quad \quad  \text{ which is the discriminant of } g_m(x).\\
&\mathbf{(C2)} \quad \text{The cubic equation } D(m) = 0 \text{ has at least two distinct real roots,}\\
&\quad \quad \quad \; \text{and at least one of them is a negative single root}.\\
&\mathbf{(C3)} \quad \text{The quadratic equation } D'(m) =0  \text{ has two distinct real roots } \\
&\quad \quad \quad \; m_- \text{ and } m_+ (> m_-) \text{ satisfying } \\
&\quad \quad \quad \quad \quad \quad \quad \quad \quad \quad \; m_- < 0,\ D(m_-) < 0,\ \text{and } D(m_+) \geq 0.\\
&\mathbf{(C4)} \quad \text{Let }\\
& \quad \quad \quad \quad \quad \, \Delta \;\; \doteq 16p^4r - 4p^3q^2 - 128p^2r^2 + 144pq^2r - 27q^4 + 256r^3, \\
& \quad \quad \quad \quad \quad \, \Delta_D \doteq 4r-p^2, \\
& \quad \quad \quad \quad \quad \, \Delta_P \doteq p,\\
& \quad \quad \quad \quad \quad \, \Delta_Q \doteq q,\\
&\quad \quad \quad \; \text{then the triplet (p, q, r) belongs to the set }\\
&\quad \quad \quad \quad \quad \quad \quad \quad  \left\{\{\Delta > 0\} \cap \bigl\{ (\Delta_D > 0) \cup(\Delta_P > 0) \bigr\} \right\}
\\
&\quad \quad \quad \quad \quad \quad \quad \quad \quad \quad \quad \quad \quad \; \mathrm{or} \text{ the set}
\\
&\quad \quad \quad \quad \quad \quad \quad \quad \;\; \{(\Delta_D=0) \cap (\Delta_P>0) \cap (\Delta_Q=0)\}. \\
&\mathbf{(C5)} \quad \text{There exists a negative }m \text{ such that } h_m(x) \geq 0 \: \text{ for any } x \in \mathbb{R} \; \text{ and that} \\ 
&\quad \quad \quad \; g_m(x) < 0 \: \text{ for any } x \in \mathbb{R}.\\
\end{align*}
\end{theorem}
\begin{proof}
$\\(\mathbf{C0} \Rightarrow \mathbf{C1})$

Consider the case $q=0$. Condition $\mathbf{(C0)}$ implies $r=f(0)>0$. The equation $D(m)=0$ has three real roots $\displaystyle -p/2-{r}^{1/2}$, $\displaystyle -p/2+{r}^{1/2}$, and $0$. If $p\geq 0$, then $\displaystyle -p/2-{r}^{1/2}$ is a negative root of $D(m)=0$. If $p<0$, $\mathbf{(C0)}$ implies 
\begin{equation*}
f(x) \geq f\left(\pm\left(-\frac{p}{2}\right)^{1/2}\right) = r-\frac{p^2}{4}>0,
\end{equation*}
and then $\displaystyle -p/2-{r}^{1/2}$ is a negative root. We choose $m_*$ such that
\begin{equation*}
m_* \in \left( -\frac{p}{2}-{r}^{1/2}, \min\left\{0, -\frac{p}{2} + {r}^{1/2}\right\} \right).
\end{equation*}
Then, $D(m_*)<0.$ \\
Consider the case $q\neq0$. 
We know 
\begin{equation}\label{eq1}
f(x) = \left(x^2+\frac{p}{2}+m\right)^2 - 2m\left(x-\frac{q}{4m}\right)^2+\frac{D(m)}{8m}.
\end{equation}
Eq. (\ref{eq1}) and $\mathbf{(C0)}$ imply that for any $m<0$, 
\begin{equation*}
    \label{eq201}
    \frac{\left(a(m)\right)^2}{(4m)^4} + \frac{D(m)}{8m} = f\left(\frac{q}{4m}\right)>0,
\end{equation*}
where $a(m)\doteq 8m^2(2m+p) + q^2$. Thus, for $m<0$, 
\begin{equation}\label{eq202}
    D(m)<-\frac{\left(a(m)\right)^2}{32m^3}.
\end{equation}
Clearly, $a(0)=q^2>0$ and $\lim_{m\to-\infty}a(m)=-\infty$. Thus, the intermediate value theorem(IVT) implies that there exists a negative $m_*$ satisfying $a(m_*)=0$. Combining this property and Eq.$(\ref{eq202})$ yields 
\begin{equation*}
    \label{eq203}
    D(m_*)<0.
\end{equation*}
$(\mathbf{C1} \Rightarrow \mathbf{C2})$

Consider the case $q=0$. The equation $D(m)=0$ has three roots $\displaystyle -p/2 - {r}^{1/2}, -p/2+{r}^{1/2}, \text{ and } 0.$ Condition $\mathbf{(C1)}$ implies that $\displaystyle-p/2-{r}^{1/2} <0$ and $r>0$. Moreover, if $-p/2+{r}^{1/2}=0$, then $D(m)=0$ has one negative root and one zero root. Thus, the equation $D(m)=0$ has at least two distinct real roots
and at least one of them is a negative single root.\\Consider the case $q\neq0$. Since $D(0)=q^2>0$ and $\lim_{m \to \infty}D(m)=-\infty$, the IVT implies there exists at least one positive single root. If there exists three positive roots, $D(m)>0$ for any $m<0$. It is a contradiction. Thus, there exists only one positive root and two negative roots. Since there exists a negative $m$ such that $D(m)<0$, the two negative roots should be distinct.
\\
$(\mathbf{C2} \Rightarrow \mathbf{C3})$

The roots of the quadratic equation $\displaystyle D'(m)=-8\left(3m^2+2pm+p^2/4-r\right)=0$ are 
\begin{equation}
\label{m_minus}
m_{-} \doteq \frac{-2p - (p^2+12r)^{1/2}}{6} \quad \text{  and }\quad \displaystyle m_{+} \doteq \frac{-2p + (p^2+12r)^{1/2}}{6}.
\end{equation}
Consider the case $q=0$. 
Condition $\mathbf{(C2)}$ says that the equation $D(m)=0$ has a negative single root $-p/2-{r}^{1/2}$.
Hence $$m_- \in \left(-\frac{p}{2}-{r}^{1/2}, \min \left\{-\frac{p}{2}+{r}^{1/2}, 0\right\}\right), \;D(m_-)<0, \;\text{ and } \; D(m_+)\geq0.$$
Consider the case $q\neq 0$. We know that the equation $D(m)=0$ has two distinct negative roots and one positive root by the IVT. Thus, $$m_-<0, \quad D(m_-)<0, \; \text{ and } \; D(m_+) \geq 0.$$
$(\mathbf{C3} \Rightarrow \mathbf{C4})$

Consider the case $D(m_+)>0$.
The condition $D(m_-)<0$ implies
\begin{align}
\label{d_m_minus}
\frac{(2p^3-72pr+27q^2) - 2(p^2+12r)^{{3}/{2}}}{27} = D(m_-) <0.
\end{align}
Hence,  
\begin{equation}\label{eq41}
 (2p^3-72pr+27q^2)-2(p^2+12r)^{{3}/{2}}<0.
\end{equation}
Since $D(m_+)>0$,
\begin{equation}\label{eq42}
(2p^3-72pr+27q^2) + 2(p^2+12r)^{{3}/{2}} > 0.
\end{equation}
Combining (\ref{eq41}) and (\ref{eq42}) yields
\begin{equation}\label{eq43}
     (2p^3-72pr+27q^2)^2 - 4(p^2+12r)^3 < 0.
\end{equation}
We can easily show that the LHS of (\ref{eq43}) equals -27$\Delta$. Thus, 
the conditions $D(m_-)<0$ and $D(m_+)>0$ together are equivalent to
\begin{equation}\label{eq44}
\Delta>0.
\end{equation}
It is trivial that
\begin{equation*}\label{eq501}
    \{m_- < 0\} = \{ ({p^2+12r})^{1/2} > -2p \} = \{ p>0 \} \cup \{ (p\leq0) \; \mathrm{and} \; (4r > p^2) \}.
\end{equation*}
Thus, 
\begin{equation}\label{eq502}
    \{m_- < 0\} = \{ (\Delta_P>0) \; \mathrm{or} \; (\Delta_D>0) \}.
\end{equation}
Consider the case $D(m_+)=0$. If $m_+<0$, then $D(m)$ is decreasing for $m>m_+$, which implies $D(0)<D(m_+)=0$. However, this contradicts the fact that $D(0)=q^2\geq0$. If $m_+>0$, then $D(m)$ is increasing for $m_-<m<m_+$. $D(0)<D(m_+)=0$. It contradicts to $D(0)\geq 0$. Therefore, it must follow that $m_+=0$. Since $m_+=0$, the negative root of $D(m)$ is $-p$. Consequently, $\Delta_P=p>0$ and $p^2=4r$, and then $\Delta_D=0$. Substituting $p^2=4r$ into the equation $D(m_+)=0$ yields $\Delta_Q=q=0$.
These conditions yield $\mathbf{(C4)}$.
\\
$(\mathbf{C4} \Rightarrow \mathbf{C5})$

Consider the case $\{(\Delta_D=0) \cap (\Delta_P>0) \cap (\Delta_Q=0)\}$. Clearly, $D(m_+)=0$ and $D(m_-)<0$ and $m_-<0$.\\
Consider the case $\{\Delta > 0\} \cap \bigl\{ (\Delta_D > 0) \; \mathrm{ or }\;(\Delta_P > 0) \bigr\}$.
From Eqs. (\ref{eq41}) $\sim$ (\ref{eq44}), we know that the condition $\Delta>0$ is equivalent to 
\begin{equation*}\label{eq51}
D(m_+)>0 \; \text{ and }\; D(m_-)<0.
\end{equation*}
Hence, the equation $D(m)=0$ has at least two distinct real roots. Eq. (\ref{eq502}) says that if either $\Delta_P>0$ or $\Delta_D>0$, then
\begin{equation}
\label{eq8}
    m_- < 0.
\end{equation}
Since the equation $D(m)=0$ has at least two distinct real roots, Eq. (\ref{eq8}) implies that the equation $D(m)=0$ has at least one negative single root. Denote its smallest root by $m_0$. The inequality $D(0)=q^2\geq0$ implies $m_- \in (m_0, 0)$ and $D(m_-)<0$. Thus
\begin{equation}
\label{eq401}
g_{m_-}(x) = 2m_-\left(x-\frac{q}{4m_-}\right)^2 - \frac{D(m_-)}{8m_-} \leq - \frac{D(m_-)}{8m_-} <0 \quad \forall x \in \mathbb{R}.
\end{equation}
Clearly, 
\begin{equation}
\label{eq402}
    h_{m_-}(x) = \left(x^2+\frac{p}{2}+m_-\right)^2 \geq 0 \quad \forall x \in \mathbb{R}.
\end{equation}
Combining Eqs. (\ref{eq401}) and (\ref{eq402}) yields $\mathbf{(C5)}$.
\\
$(\mathbf{C5} \Rightarrow \mathbf{C0})$

For a negative $m$ satisfying $h_m(x)\geq 0$ and $g_m(x)<0$ for any $x \in \mathbb{R}$,
\begin{equation*}
f(x)=h_{m}(x)-g_{m}(x)>0 \quad \forall x \in \mathbb{R}.
\end{equation*}
Thus, $\mathbf{(C0)}$ holds.
\end{proof}

The equivalence of $\mathbf{(C0)}$ and $\mathbf{(C4)}$ is well-known. See, e.g.,\ \cite{rees1922graphical}, \cite{dickson1917elementary} and \cite{lazard1988quantifier}. Up to now we use the reduced quartic polynomial. For a general quartic polynomial $ax^4+bx^3+cx^2+dx+e=0 \; (a>0)$, let $z=x+b/(4a)$, then it becomes $z^4+pz^2+qz+r=0$, with coefficients
\begin{align}\label{p}
p &= \frac{8ac-3b^2}{8a^2}, \\
\label{q}
q &= \frac{b^3-4abc+8a^2d}{8a^3}, \\
\label{r}
r &= \frac{256a^3e-64a^2bd+16ab^2c-3b^4}{256a^4}.
\end{align}
Using Eqs. (\ref{p}) $\sim$ (\ref{r}), we can apply Theorem\ref{theorem1} to a general quartic polynomial.
\section{Gram--Charlier Positivity}
\subsection{Related Works}
Historically, ensuring the positivity of the Gram--Charlier density relied on Monte Carlo simulations or numerical methods to define feasible parameter regions. See, e.g.,\ \cite{barton1952conditions}, \cite{Draper1972biometrika}, \cite{balitskaya1988biometrika}, \cite{jondeau2001gram}, and \cite{del2012gram}. Despite the significance of the analytical solutions provided by \cite{kwon2019positive, kwon2022analytic}, the formulations involve inherent algebraic complexities that hinder straightforward application.

\subsection{Valid Region}
The probability density function of the Gram--Charlier expansion is given by
$$
f(z) = P(z)\phi(z),
$$
where $\phi(z)$ denotes the standard normal density function and $$P(z)=1+\eta_3He_3(z)+\eta_4He_4(z).$$
Here, $He_3(z)=z^3-3z \text{ and } He_4(z)=z^4-6z^2+3$ are the third- and the fourth-order probabilist's Hermite polynomials, respectively. It is necessary to assume that $P(z)>0, 
\forall z \in \mathbb{R},$ i.e.,
$$P(z) = \eta_4z^4+\eta_3z^3-6\eta_4z^2-3\eta_3z+3\eta_4+1 > 0, \; \forall z\in \mathbb{R}.$$
Applying to Eqs. (\ref{p}) $\sim$ (\ref{r}), we obtain $p=-3\left(\eta_3^2+16\eta_4^2\right)\big/\left(8\eta_4^2\right)$, $q=\eta_3^3\big/\left(8\eta_4^3\right)$, and
\\
$r=-\left(3\eta_3^4-96\eta_3^2\eta_4^2-768\eta_4^4-256\eta_4^3\right)\big/\left(256\eta_4^4\right)$.
To make use of \textbf{(C4)}, we need
\begin{flalign}
&\hspace*{-0.0cm}
\Delta =
\begin{aligned}[t]
&\left(108 \eta_{3}^{6}
+ 1620 \eta_{3}^{4} \eta_{4}^{2}
+ 108 \eta_{3}^{4} \eta_{4}
- 27 \eta_{3}^{4}
+ 10368 \eta_{3}^{2} \eta_{4}^{4} \right. \\
&\quad \left.
- 288 \eta_{3}^{2} \eta_{4}^{2}
+ 27648 \eta_{4}^{6}
- 2304 \eta_{4}^{4}
+ 256 \eta_{4}^{3}
\right)\Big/ \eta_4^6 ,
\end{aligned}
&& \label{eq:Delta} \\
&\hspace*{-0.0cm}
\Delta_D =
\left(
- 3 \eta_{3}^{4}
- 48 \eta_{3}^{2} \eta_{4}^{2}
- 384 \eta_{4}^{4}
+ 64 \eta_{4}^{3}
\right)
\Big/
\left(16\eta_4^4\right),
&& \label{eq:DeltaD} \\
&\hspace*{-0.0cm}
\Delta_P =
-3
\left(\eta_{3}^2 + 16\eta_{4}^{2}\right)
\Big/
\left(8\eta_4^2\right),
&& \label{eq:DeltaP} \\
&\hspace*{-0.0cm}
\Delta_Q =
\left(\eta_{3}^{3}\right)
\Big/
\left(8\eta_4^3\right).
&& \label{eq:DeltaQ}
\end{flalign}
The $\Delta$ and $\Delta_D$ conditions make it difficult to apply \textbf{(C4)} to the positivity problem.

Instead of applying Eqs. (\ref{eq:Delta}) $\sim$ (\ref{eq:DeltaQ}) to \textbf{(C4)}, \cite{kwon2022analytic} derived the analytic positivity region through solving a cubic equation. The resulting analytical expression inherently involves nested radical functions and cosine and arc cosine terms.

By employing \textbf{(C3)}, we can present a simpler representation of the positive Gram--Charlier density. First, we obtain $m_-$ and $D(m)$ through Eqs. (\ref{m_minus}) and (\ref{d_m_minus}). Then, obtain the representation via \textbf{(C3)} as
\[
\begin{cases}
\displaystyle 0 < \eta_4 < \frac{1}{6}, & \\
\displaystyle m_- = \left(3\eta_3^2 + 48\eta_4^2 - 48^{1/2}\eta_4\left(3\eta_3^2 + 24\eta_4^2 + 4\eta_4\right)^{{1}/{2}} \right)\bigg/\left(24\eta_4^2\right) < 0, & \\
\displaystyle D(m_-) = \left(6\eta_3^2\eta_4 + \eta_3^2 + 32\eta_4^3 + 16\eta_4^2 - (27/4)^{-1/2}\left(3\eta_3^2 + 24\eta_4^2 + 4\eta_4\right)^{{3}/{2}}\right)\bigg/\eta_4^3 < 0. &
\end{cases}
\]
Since $\eta_4>0$, the inequalities become
\[
\begin{cases}
\displaystyle 0<\eta_4<\frac{1}{6}, \\
\displaystyle 3\eta_3^2+48\eta_4^2 -48^{1/2}\eta_4\left(3\eta_3^2 + 24\eta_4^2 + 4\eta_4\right)^{{1}/{2}}<0, \\
\displaystyle 6\eta_3^2\eta_4  + \eta_3^2 + 32\eta_4^3  + 16\eta_4^2 - {(27/4)^{-3/2}}\left(3\eta_3^2 + 24\eta_4^2 + 4\eta_4\right)^{{3}/{2}}<0.
\end{cases}
\]
The intersection of above conditions is depicted in Fig. \ref{fig:gram-charlier}.
\begin{figure}[!ht]
    \centering    \includegraphics[width=0.7\textwidth]{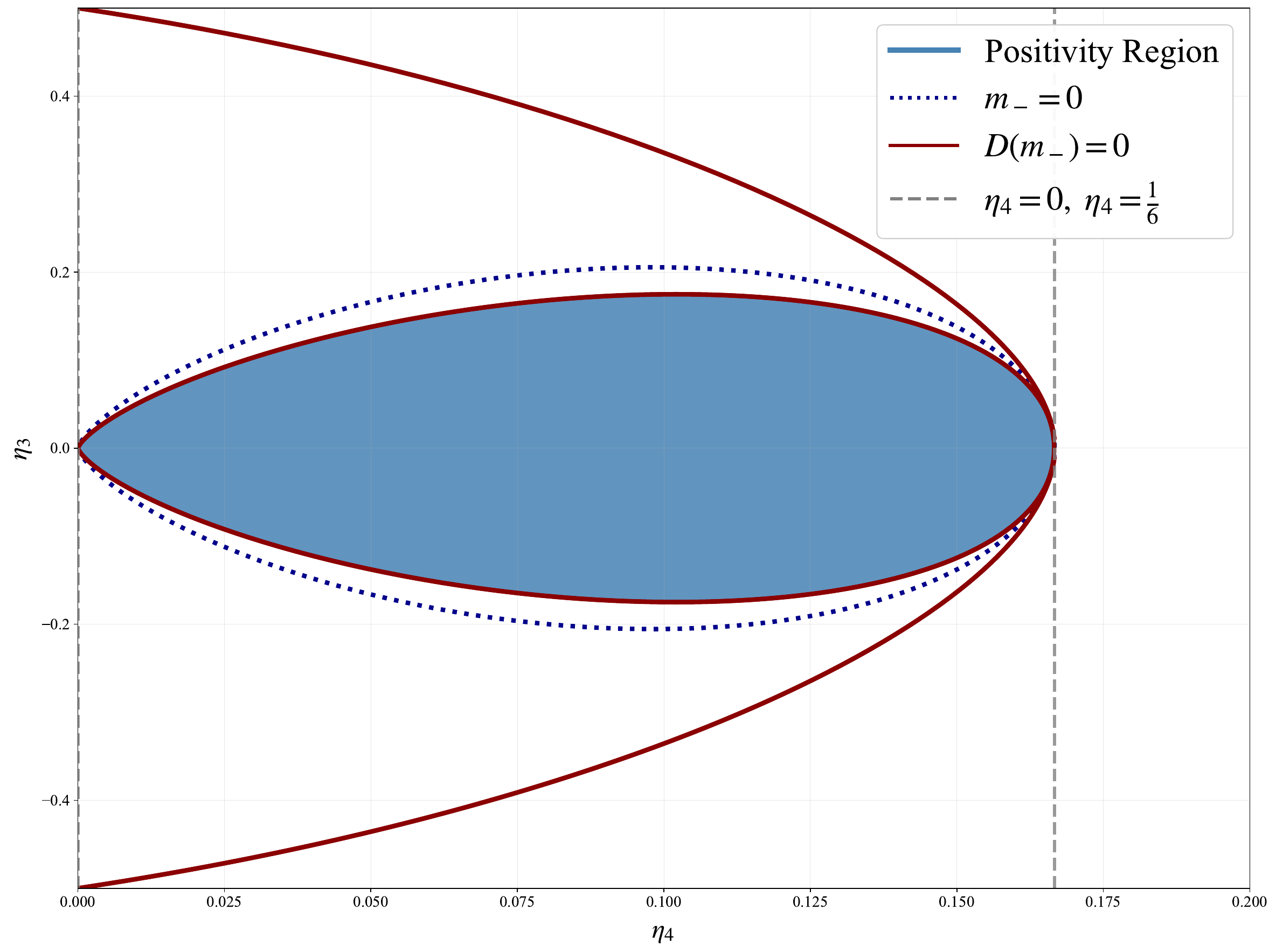}
    \caption{Gram--Charlier positivity region.
    \\
\textbf{Alt text}:The plot shows the valid parameter region in the 
    $(\eta_4,\eta_3)$ plane. The region is symmetric about $\eta_3=0$ and is bounded by the curves corresponding to 
    $m_-=0$ and $D(m_-)=0$.
    }
    \label{fig:gram-charlier}
\end{figure}
This representation is substantially more concise than the one by \cite{kwon2022analytic}, since the former eliminates the need for nested radicals and trigonometric functions appearing in the latter.
Therefore, the valid region is given by the above inequalities, together with the boundary point $(\eta_3,\eta_4)=(0,0)$, which corresponds to the standard normal density.

\section*{Acknowledgement}
Jung Chan Lee is also affiliated with the Institute of Medical and Biological Engineering, Medical Research Center, Seoul National University, Seoul 03080, Republic of Korea.
\bibliographystyle{preprint}
\bibliography{paper-ref}
\end{document}